\documentclass[12pt]{amsart}

\usepackage{amssymb, amsmath, hyperref}

\DeclareMathOperator{\rank}{rank}
\DeclareMathOperator{\co}{co}
\DeclareMathOperator{\face}{face}

\DeclareMathOperator{\tr}{tr}

\def\C{\mathbb{C}}

\def\B{\mathcal{B}}

\newtheorem{theorem}{Theorem} 
\newtheorem{lemma}[theorem]{Lemma} 
\newtheorem{corollary}[theorem]{Corollary}

\theoremstyle{definition}
\newtheorem*{notation}{Notation}
\newtheorem*{definition}{Definition}
\theoremstyle{remark}
\newtheorem*{example}{Example}
\newtheorem*{remark}{Remark}

\newif\ifproofing
\proofingtrue
\newcommand{\note}[1]{\ifproofing\marginpar{#1}\fi}

\def\<{\langle}
\def\>{\rangle}

\def\wT{{\widetilde T}}

\DeclareMathOperator{\im}{im}

\DeclareMathOperator{\Sep}{Sep}

\begin{document}

\title{Finding decompositions of a class of separable states}
\author{Erik Alfsen}
\email{alfsen@math.uio.no}
\address{Mathematics Department, University of Oslo, Blindern 1053, Oslo, Norway}
\author{Fred Shultz}
\email{fshultz@wellesley.edu}
\address{Mathematics Department, Wellesley College, Wellesley, Massachusetts 02481, USA}
\date{February 16, 2012}

\begin{abstract}
By definition a  separable state  has the form $\sum_{i=1}^p A_i \otimes B_i$, where $0 \le A_i, B_i$ for each $i$. In this paper we consider the class of states which admit such a decomposition with $B_1, \ldots, B_p$ having independent images. We give a simple intrinsic characterization of this class of states, and  starting with a density matrix in this class, describe a procedure to find such a decomposition with $B_1, \ldots, B_p$ having independent images, and  $A_1, \ldots, A_p$ being distinct with unit trace.  Such a decomposition is unique, and we relate this to the facial structure of the set of separable states.
 
A special subclass of such separable states are those for which the rank of the matrix matches one marginal rank.  Such states have arisen in previous studies of separability (e.g., they are known to be a class for which the PPT condition is equivalent to separability). 

The states investigated also include a class that corresponds (under the Choi-Jamio\l kowski isomorphism) to the quantum channels called quantum-classical and classical-quantum by Holevo.
\end{abstract}

\maketitle

\section*{Introduction}

  A state on $M_m \otimes M_n$ is separable if it is a convex combination of product states.    States that are not separable are said to be entangled and are of substantial interest in quantum information theory since entanglement is at the heart of many applications.  Some useful necessary conditions are known for separability, e.g., the PPT condition, by which a separable state must have positive partial transpose \cite{Peres}.   There also are some necessary and sufficient conditions, e.g. \cite{Horo3}, which however are difficult to apply. Thus it would be of great interest to find a practical test for separability, at least for a significant class of states.  
  
  Closely related to this is the goal of finding a procedure to decompose interesting classes of separable states into a convex combination of  product states.  Such a procedure would not only shed light on separable states, but would provide a separability test for that class.
  
We will identify states with their associated density matrix, and also consider unnormalized states, which are then associated with positive semi-definite matrices.
(We will abbreviate ``positive semi-definite" to simply ``positive" hereafter.)
Thus a density matrix $T\in M_m \otimes M_n$ is separable if it admits a representation
\begin{equation}\label{sep def} T = \sum_{\gamma=1}^p A_\gamma \otimes B_\gamma\end{equation}
where each $A_\gamma$ and $B_\gamma$ is positive. Such a density matrix $T$ represents a mixed state on a bipartite quantum system composed of two subsystems, the $A$-system and the $B$-system, associated with $M_m$ and $M_n$ respectively.

In our previous paper \cite{JMP}, the authors studied  separable states with such a representation with each $A_\gamma$ and $B_\gamma$ rank one,  with the  requirement that  $B_1, \ldots, B_p$ be projections onto linearly independent vectors.
This class of states turns out to be same as the set of separable states $T$ with the property that  $T$ and the marginal state $\tr_AT = T_B$ obtained by tracing out over the $A$-system have the same rank, cf. Lemma \ref{rank lemma}.  The equivalence of these two formulations was established for states on $M_2 \otimes M_n$ in \cite{Kraus}, and then in complete generality in \cite[Lemma 6, and proof of Thm. 1]{HLVC}, where it was also shown that for states satisfying this rank requirement, the PPT condition is equivalent to separability. (An alternate  proof of the equivalence of these rank and independence conditions was given in \cite[Lemma 13]{RuskaiWerner}.)  In \cite{HLVC} the authors also gave a procedure for decomposing such states into a convex sum of pure product states, based on an inductive argument for finding a certain kind of product basis, and then a reduction to a block matrix whose blocks are normal and commute.    In this paper we also make use of a reduction to this type of matrix. The existence of special families of commuting normal matrices  played an important role in the investigation of separability in \cite{Lewenstein} as well.
  
The current paper investigates separable states for which there is no rank restriction, but   admitting a representation \eqref{sep def} in which $B_1, \ldots, B_p$ have independent images.  We call such states $B$-independent, and give an intrinsic way to determine if a state falls in this category (Theorem \ref{main}). We show that without knowing an explicit decomposition to begin with, there is a canonical way to locally filter $T$ to yield a state $\wT$ which admits a representation \eqref{sep def} in which $B_1, \ldots, B_p$ are orthogonal. (Of course, there is nothing special about the $B$-system compared to the $A$-system, and all results in this paper are valid with the roles of the $A$ and $B$ systems interchanged.) 

 This is then used to give a canonical form for $T$, and to find a decomposition of $T$ of the form \eqref{sep def}, cf. Theorem \ref{main}.  This decomposition can be chosen so that $A_1, \ldots, A_p$ are distinct and have unit trace, and in that case the representation is unique.  It is then simple to decompose further to get a representation of $T$ as a convex combination of pure product states (i.e., of density matrices where each is the projection onto the span of a product vector), and we describe when this decomposition is unique (Theorem \ref{unique}.) Finally, we show in Theorem \ref{faces} that if a state has a representation \eqref{sep def} with the images of the $A_\gamma$ disjoint and the images of the $B_\gamma$ independent, then the face of the space $S$ of separable states that is generated by this state is the direct convex sum of separable state spaces of lower dimension.

The density matrices investigated here are closely related to interesting classes of completely positive maps.  A  completely positive  map $\Phi:M_m \to M_n$ is \emph{entanglement breaking} if $(I \otimes \Phi)(\Gamma)$ is separable for all positive $\Gamma$, cf. \cite{HSR, Ruskai}. The Choi-Jamio\l kowski isomorphism \cite{Choi, Jamiol} is a linear isomorphism under which completely positive maps correspond to positive matrices. Under this correspondence, entanglement breaking maps correspond to separable matrices, so the results of this paper on convex decompositions of a class of separable states then can be transferred to give information about decompositions and identification of the corresponding class of entanglement breaking maps.

In particular, there are two important classes of entanglement breaking maps (quantum-classical channels and classical-quantum channels)  that have Choi matrices in the class of separable states investigated in the current paper.  These classes were originally singled out by Holevo \cite{Holevo}, and further investigated as special cases of entanglement breaking maps by Horodecki, Shor, and Ruskai in \cite{HSR, Ruskai}. These are shown in Theorem \ref{CQ} to be special cases of the classes of $A$-orthogonal and $B$-orthogonal density matrices, which play a key role in the current paper.  Theorem \ref{orthogonal} and Theorem \ref{CQ} together provide an intrinsic way to identify such  quantum channels without knowing an explicit Kraus decomposition ahead of time, as well as giving a procedure to find a Kraus decomposition of the appropriate form.

\section*{A class of separable density matrices}

\begin{definition} Subspaces $V_1, \ldots, V_p$ of a vector space are \emph{independent} if their sum is a direct sum. This is equivalent to  the implication
$$\sum_{\gamma = 1}^p x_\gamma = 0 \text{ with $x_\gamma \in V_\gamma$  for $1 \le \gamma \le p$} \implies \text{ all $x_\gamma = 0$}. $$
\end{definition}

We now define the central  class of separable density matrices that we will investigate.  Later in Theorem \ref{main} we will give an intrinsic characterization of this class.

\begin{definition}
A density matrix $T \in M_m \otimes M_n$ is \emph{$B$-independent} if $T$ admits a decomposition
\begin{equation}\label{T}
T = \sum_{\gamma =1}^p  A_\gamma \otimes B_{\gamma}
\end{equation}
where $0 \le A_\gamma, B_\gamma$ for $1\le \gamma \le p$, with the images of $B_1, \ldots, B_p$ independent. 
\end{definition}

\begin{example} Let $x_1, \ldots, x_p \in \C^m$ and $y_1, \ldots, y_p \in \C^n$ be unit vectors,
with $y_1, \ldots, y_p$ linearly independent, and $0 < \lambda_1, \ldots, \lambda_p$ with $\sum_\gamma \lambda_\gamma = 1$.
Let $T$ be the convex combination
\begin{equation}\label{MR} T = \sum_{\gamma=1}^p \lambda_\gamma P_{x_\gamma} \otimes P_{y_\gamma}
\end{equation}
where for a unit vector $z$, $P_z$ denotes the projection onto $\C z$.
Then $T$ is $B$-independent.  The uniqueness of such decompositions, and the facial structure of faces of the separable state space  generated by such  states, were investigated by the current authors in \cite{JMP}. As was discussed in the introduction, such states  played an important role in \cite{HLVC} and also appeared in \cite{RuskaiWerner}.
\end{example}


We will return to the subject of $B$-independent states after developing some necessary results.

\section*{One sided local filtering}

\begin{definition} A linear map $\Phi:M_d \to M_d$ is a \emph{filter} if there is a positive $A \in M_d$ such that $\Phi(X) = AXA$.  (We do not require that $A$ be invertible.)
 A map $\Phi:M_m \otimes M_n \to M_m \otimes M_n$ is a \emph{local filter} if there are positive $A, B$ such that $\Phi(X) = (A \otimes B)X(A \otimes B)$. 
\end{definition}

Applications of filtering, e.g., to distillation of entanglement, date back at least to \cite{Bennett, Bennett1, Gisin}. It is well known that we can  apply a local filter to any density matrix to arrange for one or the other partial trace to be a projection, as we now describe. 

\begin{definition} If $A \ge 0$, then $A^\#$ denotes the Penrose pseudo-inverse of $A$, i.e., the unique positive matrix which is zero on $(\im A)^\perp$ and satisfies  $A^\# A = AA^\# = P_A$, where $P_A$ is the projection onto the image of $A$. If a spectral decomposition of $A$ is $A = \sum_i \lambda_i P_i$ with all $\lambda_i > 0$, then $A^\# = \sum_i \lambda_i^{-1} P_i$.  
\end{definition}

\begin{definition} We write $\tr_B$ and $\tr_A$ for the partial trace maps on $M_m \otimes M_n$, and if $T \in M_m \otimes M_n$ then we write $T_B = \tr_A T$ and $T_A = \tr_B T$.
\end{definition}

\begin{definition} Let $0 \le T \in M_m \otimes M_n$. Then we denote by $\wT$ the matrix
\begin{equation}\label{wT def}
\wT = (I \otimes ((T_B)^\#)^{1/2})T(I \otimes ((T_B)^\#)^{1/2}).
\end{equation}
\end{definition}

We view the pair $(\wT, T_B)$ as partitioning information about $T$ into a state $T_B$ that contains information about $T$ on the subsystem $B$, and another part $\wT$ which contains  information about $T$ relating to the system $A$ as well as the interaction between $A$ and $B$ systems.  

We will show later in this section that $T$ can be recovered from the pair $(\wT, T_B)$.
First, we  discuss  various facts about partial traces and filters which we need.

If $T = \sum_{ij} E_{ij} \otimes T_{ij}\in M_m \otimes M_n$, then by definition
\begin{equation} \label{pt def} \tr_A T = \sum_i T_{ii} \text{ and }  \tr_B T = \sum_{ij} \tr(T_{ij}) E_{ij}.
\end{equation}

It is well-known that the partial trace maps are positive maps.  We now show that they are also \emph{faithful}, i.e., if $T \ge 0$ and either partial trace of $T$ is zero, then $T$ is zero.    
(We expect the following is well-known, but we have  included it here for lack of an explicit reference.)

\begin{lemma}\label{faithful}
The partial trace maps are faithful.
\end{lemma}

\begin{proof} 
Let  $0 \le T = \sum_{ij} E_{ij} \otimes T_{ij}\in M_m \otimes M_n$. We first show
\begin{equation}\label{Aii} T = 0 \iff \text{ $T_{ii} = 0$ for $1 \le i \le m$.}
\end{equation}
If $T = 0$, then clearly all $T_{ij} = 0$, so in particular all $T_{ii} = 0$. Conversely, suppose all $T_{ii}$ are zero. 
Then $(E_{ii} \otimes I)T(E_{ii} \otimes I) = E_{ii} \otimes T_{ii} = 0$ for all $i$. Then for $1 \le i\le m$,
$$((E_{ii}\otimes I)T^{1/2})((E_{ii}\otimes I)T^{1/2})^* = 0$$
so $(E_{ii}\otimes I)T^{1/2} = 0$. Adding gives
$$0 = \sum_i (E_{ii} \otimes I) T^{1/2} =  (I \otimes I)T^{1/2} = T^{1/2}$$
and hence $T = 0$.

Now we are ready to prove faithfulness of the partial traces.   By \eqref{pt def}, if $\tr_B(T) =0$, then in particular $\tr(T_{ii}) = 0$ for each $i$. Since $0 \le T$, then $0 \le T_{ii}$ for each $i$, so $\tr(T_{ii}) = 0$ implies $T_{ii} = 0$ for all $i$ and thus $T = 0$.

On the other hand,  $\tr_A(T) = 0$ implies $\sum_i T_{ii} = 0$, and by positivity of each $T_{ii}$, we again have $T_{ii} = 0$ for each $i$, and thus $T = 0$.
\end{proof}

We next review some useful facts about projections and images. (For additional background, cf. \cite[Chapter 3]{Alfsen-Shultz}.) If $A = A^*$ and $P$ is a projection, then
\begin{equation}\label{range}
\im A \subset \im P \iff PAP = A.
\end{equation}
(Indeed,  if $\im A \subset \im P$, then $PA = A$, so taking adjoints and using $A^*= A$ gives $A = AP$. Then $PAP = P(AP) = PA= A$. The converse implication is clear.)

If $E \in M_r$ is a projection, then we write $E' = I-E$, where $I$ is the identity in $M_r$.  Note that for $E$ a projection in $M_m$, $(E \otimes I)' = (I \otimes I)- (E \otimes I) = E' \otimes I$. 
For any projection $R$ and positive operator $T$ we have
\begin{equation}\label{prime}
RTR = T \iff R'TR' = 0,
\end{equation}
cf., e.g., \cite[Lemma 2.20]{Alfsen-Shultz}.

Finally, we observe that if $0 \le A_1, A_2, \ldots, A_p$, then
\begin{equation}\label{image of sum}
\im \sum_i A_i = \sum_i \im A_i.
\end{equation} 
Indeed, for $1\le j\le p$ we have $A_j \le \sum_i A_i$ so $\ker \sum_i A_i \subset \ker A_j$. Taking orthogonal complements shows $\im A_j \subset \im \sum_iA_i$, which implies $\sum_j \im A_j \subset \im \sum_i A_i$. The opposite containment is evident, so   \eqref{image of sum} follows.

 The following result is clear for separable $T$, but requires a little more work for general $T$.

\begin{lemma}\label{product range} If $0 \le T \in M_m \otimes M_n$, the minimal product subspace containing the image of $T$ is $\im T_A \otimes \im T_B$.  In particular, if $P_B$ is the projection onto the image of $T_B$, then $(I \otimes P_B)T(I \otimes P_B) = T$.
\end{lemma}


\begin{proof} Let $V\subset\C^m$ and $W \subset \C^n$ be subspaces, and let the corresponding projections be $P$ and $Q$. Then by \eqref{range}, $\im T \subset V \otimes W$ iff $(P\otimes Q)T(P \otimes Q) = T$.

Note $(P \otimes Q)T(P \otimes Q) = T$ is equivalent to the combination of $(P \otimes I)T(P \otimes I) = T$ and $(I \otimes Q)T(I \otimes Q) = T$.  Thus it suffices to show that 
\begin{equation}\label{range tensor}
(P \otimes I)T(P \otimes I) = T \iff \im P \supset \im(T_A)
\end{equation}
 together with the corresponding statement for $T_B$.  Since the proof of the  statements for $T_A$ and $T_B$ are essentially the same,  we  just will prove the statement for $T_A$.

   We will make use of the following identity valid for all $T \in M_m \otimes M_n$ and all $X \in M_n$:
\begin{equation} \label{pt}
\tr_B(X \otimes I)T (X \otimes I) = X(\tr_B T) X.
\end{equation}
Thus 
\begin{align}
(P \otimes I)T (P \otimes I) = T &\iff (P' \otimes I)(T)(P' \otimes I) = 0 \text{ by \eqref{prime}}\cr
&\iff \tr_B((P' \otimes I)T(P' \otimes I)) =0\text{ by Lemma \ref{faithful}}\cr
& \iff  P'(\tr_B T) P' = 0 \text{ by \eqref{pt}} \cr
& \iff P\tr_B T P = \tr_BT \text{ by \eqref{prime}} \cr
& \iff \im P \supset \im \tr_B T = \im(T_A) \text{ by \eqref{range}}.
\end{align}
This completes the proof of \eqref{range tensor}, and hence finishes the proof of the lemma.
\end{proof}

The next result relates properties of $T$ and $\wT$, and shows that $T$ can be recovered from the pair $(\wT, T_B)$.

\begin{lemma} \label{wT lemma}Let $0 \le T \in M_m \otimes M_n$ and define $\wT$ as in \eqref{wT def}. Then
\begin{equation}\label{sep}
T = (I \otimes T_B^{1/2})\wT(I \otimes T_B^{1/2}).
\end{equation}
 $T$ will be separable iff $\wT$ is separable, and $\tr_A \wT = P_B$ (where $P_B$ is the projection onto the image of $T_B$).

\end{lemma}

\begin{proof} From the definition of $\wT$, separability of $T$ implies that of $\wT$. 
For $F = ((T_B)^\#)^{1/2}$ we have
$$\tr_A \wT = \tr_A (I \otimes F)T(I \otimes F) = F(\tr_A T)F = FT_BT = P_B.$$
 Furthermore, 
\begin{align}\label{reverse}
&(I \otimes T_B^{1/2})\wT(I \otimes T_B^{1/2}) \cr
&=
 (I \otimes T_B^{1/2}) (I \otimes ((T_B)^\#)^{1/2})T(I \otimes T_B^{1/2})(I \otimes ((T_B)^\#)^{1/2}) \cr
 &= (I\otimes P_B)T(I \otimes P_B)
 \end{align}
This would prove \eqref{sep} if we knew the range of $T$ were contained in $\C^m \otimes \im P_B$. This follows from Lemma \ref{product range}. Finally \eqref{sep} shows that separability of $\wT$ implies separability of $T$.
\end{proof}

\section*{$B$-orthogonal density matrices}

In this section we describe a canonical form for a class of positive matrices which we call $B$-orthogonal, and which is a subclass of the $B$-independent matrices. In the following section we will apply these results to achieve a canonical representation for the full class of $B$-independent matrices.

\begin{definition} Positive matrices in $M_r$ are \emph{orthogonal} if their images are orthogonal.
 A density matrix $T$ is \emph{$B$-orthogonal} if it admits a representation
\begin{equation}\label{T orthog}
T = \sum_{\gamma=1}^p A_\gamma \otimes B_\gamma
\end{equation}
with  $0 \le A_\gamma, B_\gamma$ and with the $\{B_\gamma\}$ matrices orthogonal. 
Similarly we say $T$ is \emph{$A$-orthogonal} if it admits  a representation \eqref{T orthog} with  the $\{A_\gamma\}$ matrices orthogonal.
\end{definition}

\begin{definition}  $\{E_{ij}\}$ are the standard matrix units of $M_m$. 
For any matrix $S \in M_m\otimes M_n$,  we denote by $\{S_{ij} \mid 1 \le i, j \le m\}$ the unique matrices in $M_n$ such that 
$$S = \sum_{ij} E_{ij} \otimes S_{ij}.$$ 
\end{definition}

The following gives a canonical form for $B$-orthogonal matrices, and a readily tested necessary and sufficient condition for $B$-orthogonality.

\begin{theorem} \label{orthogonal}Let $0 \le T = \sum_{ij} E_{ij} \otimes T_{ij} \in M_m \otimes M_n$. Then the following are equivalent.
\begin{enumerate}
\item[(i)] $T$ is $B$-orthogonal.
\item[(ii)] All $T_{ij}$ are normal and mutually commute.   
\end{enumerate}
Furthermore, if $T$ is $B$-orthogonal, then $T$ admits a unique representation
\begin{equation}\label{T5}
T = \sum_{\gamma=1}^p A_\gamma \otimes Q_\gamma,
\end{equation}
with  $Q_1, \ldots, Q_p$ orthogonal projections, and $A_1, \ldots, A_p$ distinct nonzero positive matrices. 

The projections $Q_1, \ldots, Q_p$ will be the projections onto the joint eigenspaces of $\{T_{ij}\}$ (excluding the joint zero  eigenspace), and will have sum $P_B$ (the projection onto the image of $T_B$). The matrices $A_\gamma$ are given by
\begin{equation}\label{Agamma} A_\gamma = \tr_A (I \otimes Q_\gamma)T(I \otimes Q_\gamma).
\end{equation}
\end{theorem}

\begin{proof}(i) $\implies$ (ii). If \eqref{T orthog} holds with $B_1, \ldots, B_p$  orthogonal, then for each pair of indices $i,j$
\begin{equation}\label{Tpq}
T_{ij} = \tr_A ((E_{ji} \otimes I)T) = \sum_{\gamma=1}^p \tr(E_{ji}A_\gamma) B_\gamma.
\end{equation}
Since $B_1, \ldots, B_p$ are orthogonal, then   $B_1, \ldots, B_p$ commute. It follows that the matrices $\{T_{ij}\mid 1 \le i,j\le m\}$ commute and are normal. 

(ii) $\implies$ (i) Conversely, suppose $\{T_{ij}\mid 1 \le i,j\le m\}$ commute and are normal.  Define $Q_1, \ldots, Q_p$ to be the projections onto the joint eigenspaces (for non zero eigenvalues) of $\{T_{ij}\}$. For each $i, j$ write 
\begin{equation}\label{Tij} T_{ij} = \sum_{\gamma=1}^p \lambda_\gamma^{i,j}Q_\gamma.\end{equation}
Then
\begin{align}
T &= \sum_{ij} E_{ij} \otimes T_{ij} = \sum_{ij} E_{ij} \otimes \left(\sum_{\gamma=1}^p \lambda_\gamma^{i,j} Q_\gamma\right)\cr
&= \sum_{\gamma=1}^p \left(\sum_{ij} \lambda_\gamma^{i,j} E_{ij} \right) \otimes Q_\gamma.
\end{align}

For each $\gamma$ define $A_\gamma = \sum_{ij} \lambda_\gamma^{i,j} E_{ij} \in M_m$. 
Then $T = \sum_\gamma A_\gamma \otimes Q_\gamma$.  For each $i,j,
\gamma$ we have $\lambda_\gamma^{i,j} = (A_\gamma)_{ij}$. Thus by the definition of the joint eigenspaces of $\{T_{ij}\}$, for $\gamma_1 \not= \gamma_2$ we must have $A_{\gamma_1} \not= A_{\gamma_2}$, and hence $A_1, \ldots, A_p$ are distinct.
Now orthogonality of $Q_1\ldots, Q_p$ implies \eqref{Agamma}.  

Finally, we prove uniqueness.  Suppose that we are given any representation \eqref{T5} of $T$ where $\{Q_\gamma\}$ are orthogonal projections and $\{A_\gamma\}$ distinct nonzero positive matrices. Then for $1 \le i,j \le m$,
$$T_{ij} = \tr(E_{ji} \otimes I)T = \sum_\gamma \tr(E_{ji}A_\gamma) Q_\gamma.$$
Then the image of each $Q_\gamma$ consists of eigenvectors for  $T_{ij}$ for the eigenvalues $\tr(E_{ji}A_\gamma)$, and by distinctness of $A_1, \ldots, A_p$ for $\gamma_1 \not= \gamma_2$ there is some pair of indices $i,j$ such that $\tr(E_{ji}A_{\gamma_1}) \not= \tr(E_{ji}A_{\gamma_2})$, so the $Q_\gamma$ are precisely the projections onto the joint eigenspaces.  
\end{proof}

\begin{remark} The condition (ii) is equivalent to the existence of an orthonormal basis of joint eigenvectors for $\{T_{ij}\}$, as is well known.
\end{remark}

\section*{A canonical form for $B$-independent matrices}

The following describes how to map  positive matrices with independent images to orthogonal projections by filtering with a positive matrix.  We say an Hermitian matrix $A \in M_n$ \emph{lives} on a subspace $H$ of $\C^n$ if $\im A \subset H$ (or equivalently, if $A = 0$ on $H^\perp$).

\begin{lemma}\label{projections} Let $X_1, \ldots, X_p \in \C^n$ with $\im X_1, \ldots, \im X_p$ independent, and let $P$ be the projection on the image of $\sum_i X_i$. Then $$A = ((\sum_i X_i)^\#)^{1/2}$$ is the unique positive matrix living on $\im P$ such that $\{AX_iA\mid 1 \le i \le p\}$ are  orthogonal projections with sum $P$. \end{lemma}

\begin{proof} Let $A = ((\sum_i X_i)^\#)^{1/2}$, and define $Y_i = AX_iA$ for $1\le i \le p$. Then
$$\sum_i Y_i = \sum_i AX_iA = A\left(\sum_i X_i \right)A = P,$$
where $P$ is the projection onto the image of $\sum_i X_i$.  By assumption, $\im X_1, \ldots, \im X_p$ are independent.  Since for each $i$, $A$ is invertible on $\im P \supset \im X_i$, and $\im Y_i \subset A(\im X_i)$, then $Y_1, \ldots, Y_p $ have independent images. Now for $1 \le j \le p$,
$$Y_j = PY_j = \sum_i Y_iY_j,$$
and then independence of the $Y$'s implies $Y_iY_j = 0$ for $i \not= j$, and $Y_j^2 = Y_j$, so  $Y_1, \ldots, Y_p$ are orthogonal projections with sum $P$.

Finally, to prove uniqueness, suppose that $0 \le A_0$, with $A_0$ living on $\im P$ and with $\{A_0X_iA_0\mid 1 \le i \le p\}$ projections with sum $P$.  Then $\im A_0 \subset \im P$, and
\begin{equation}\label{A0}
A_0(\sum_i X_i) A_0 = P,
\end{equation}
so $\im A_0 = \im P$.  Multiplying \eqref{A0} by $A_0^{\#}$ on left and right of each side gives $\sum_i X_i = (A_0^\#)^2$, so $A_0 = ((\sum_i X_i)^\#)^{1/2}$.
\end{proof}

\begin{theorem} \label{main} Let $0 \le T\in M_m \otimes M_n$.  The following are
equivalent.
\begin{enumerate}
\item[(i)] $T$ is $B$-independent.
\item[(ii)] $\wT$ is $B$-orthogonal.
\item[(iii)] All $\wT_{ij}$ are normal and mutually commute.   
\end{enumerate}
If $T$ is $B$-independent then $T$ admits a unique decomposition
  \begin{equation}\label{T2}
 T = \sum_{\gamma=1}^p A_\gamma \otimes B_\gamma
 \end{equation}
with  $0 \le A_\gamma, B_\gamma$, $\tr A_\gamma = 1$, $A_1, \ldots, A_p$ distinct, and $B_1, \ldots, B_p$ independent.

Let $Q_1, \ldots, Q_p$   be the projections corresponding to the joint eigenspaces of $\{\wT_{ij}\}$ excluding the subspace corresponding to the zero eigenvalue. Then the unique decomposition \eqref{T2} is given by 
 \begin{equation}\label{Bgamma}
 B_\gamma = (T_B)^{1/2}Q_\gamma (T_B)^{1/2}
 \end{equation}
   and
  \begin{equation}\label{Agamma2}
  A_\gamma = \tr_B(I \otimes Q_\gamma)\wT(I \otimes Q_\gamma),
    \end{equation}
and the sum of the projections $Q_\gamma$ will be the projection onto the image of $T_B$.
\end{theorem}

\begin{proof} If $T$ is $B$-independent, then by definition there are positive matrices  $A_1, \ldots, A_p$ and positive matrices $B_1, \ldots, B_p$ with independent images such that  
\begin{equation} \label{T formula}
T = \sum_{\gamma=1}^p A_\gamma \otimes B_\gamma.
\end{equation}
If necessary, we absorb scalar factors into the $B_\gamma$ so that $\tr A_\gamma = 1$ for all $\gamma$, and we combine terms if necessary so that $A_1, \ldots, A_p$ are distinct.

Now by the definition \eqref{wT def} of $\wT$,
\begin{equation}\label{wT6}
\wT = \sum_{\gamma=1}^p A_\gamma \otimes ((T_B)^\#)^{1/2}B_{\gamma} ((T_B)^\#)^{1/2} = \sum_{\gamma=1}^p A_\gamma \otimes Q_\gamma
\end{equation}
where 
\begin{equation} \label{Qgamma}
Q_\gamma = ((T_B)^\#)^{1/2}B_\gamma ((T_B)^\#)^{1/2}.
\end{equation}
By Lemma \ref{projections}, since $T_B = \sum_\gamma B_\gamma$, then $Q_1, \ldots, Q_p$ are projections with sum the projection onto the image of $\sum_\gamma B_\gamma$, and hence $\sum_\gamma Q_\gamma = P_B$.  Thus $\wT$ is $B$-orthogonal.  Furthermore, by the uniqueness statement of Theorem \ref{orthogonal}, $Q_\gamma$ and $A_\gamma$ must be as described in that Lemma (with $\wT$ in place of $T$). By \eqref{Qgamma}, since each $B_\gamma$ has range contained in the range of $T_B$, then multiplying \eqref{Qgamma} on both sides by $T_B^{1/2}$ gives \eqref{Bgamma}, and \eqref{Agamma2} follows either from \eqref{wT6} or from Theorem \ref{orthogonal}. Thus we have shown that if $T$ is $B$-independent, then $T$ admits a unique representation as specified in the theorem.

To show that $B$-orthogonality of $\wT$ implies $B$-independence of $T$, we apply Theorem \ref{orthogonal} again. We have the representation
$$\wT = \sum_\gamma A_\gamma \otimes Q_\gamma$$
where $A_\gamma$ and $Q_\gamma$ are defined as in Theorem \ref{orthogonal} with $\wT$ in place of $T$. Note that the image of each $Q_\gamma$ will be contained in the image of $\wT_B$, and $\wT_B= P_B$ by Lemma \ref{wT lemma}.  Define $B_1, \ldots, B_p$ by \eqref{Bgamma}. Orthogonality of the $Q_\gamma$ implies that their images are independent. By definition,  $(T_B)^{1/2}$ is invertible on the range of $T_B$, and $\im B_\gamma \subset (T_B)^{1/2}(\im Q_\gamma)$, so    $B_1, \ldots, B_p$ have independent images and are positive. By equation \eqref{sep} of Lemma \ref{wT lemma},
$$T = (I \otimes T_B^{1/2})\wT(I \otimes T_B^{1/2}) = \sum A_\gamma \otimes B_\gamma,$$
so $T$ is $B$-independent.

Finally, equivalence of (ii) and (iii) follows from Theorem \ref{orthogonal}.

\end{proof}

\section*{Connections with QC and CQ quantum channels}

\medskip

We will show in this section that the quantum channels known as  classical-quantum channels and quantum-classical channels correspond under the Choi-Jamio\l kowski isomorphism to density matrices that are in the classes of matrices we have called $A$-orthogonal or $B$-orthogonal respectively. The remainder of this paper is independent of this section.

\begin{definition} Let $\Phi: M_m \to M_n$ be a quantum channel (i.e., a completely positive trace preserving map).  If it is possible to choose $0 \le F_1, \ldots, F_q \in M_m$, $0 \le R_1, \ldots, R_q$, and $\tr R_k=1$ for all $k$
such that 
\begin{equation}\label{Holevo}
\Phi(X) = \sum_{k=1}^q \tr(F_k X) R_k
\end{equation}
such a representation is called a \emph{Holevo form} for $\Phi$. (Note that since $\Phi$ is assumed to be trace preserving, we must have $\sum_k F_k = I$.)

\end{definition}
  
The following notion is due to Holevo \cite{Holevo}, and was further investigated in \cite{HSR} in the context of entanglement breaking maps. 

\begin{definition} A quantum channel $\Phi: M_m \to M_n$ is a \emph{classical-quantum} (CQ) channel if $\Phi$ admits a Holevo form \eqref{Holevo} with $F_1, \ldots, F_q$ rank one projections (necessarily with sum $I_m$ since $\Phi$ is a quantum channel). Similarly, one says $\Phi$  is a \emph{quantum-classical} (QC) channel if $\Phi$ admits a Holevo form with $R_1, \ldots, R_q$ rank one projections with sum $I_n$.
\end{definition}

\begin{definition} If $\Phi:M_m \to M_n$ is a linear map, the associated \emph{Choi matrix} is the matrix in $M_m \otimes M_n$ defined by
$$C_\Phi = \sum_{ij} E_{ij} \otimes \Phi(E_{ij}),$$
where $\{E_{ij}\}$ are the standard matrix units of $M_m$.
\end{definition}

It was shown by Choi \cite{Choi} that $\Phi$ is completely positive iff $C_\Phi$ is positive semi-definite. Note that $\Phi$ will be trace preserving iff $\tr\Phi(E_{ij}) = \delta_{ij}$, or equivalently, iff $ \tr_B C_\Phi = I$.


\begin{lemma} \label{Choi} Let $F_1, \ldots, F_q \in M_m$ and $R_1, \ldots, R_q \in M_n$. Define $\Phi:M_m \to M_n$ by 
$$\Phi(X) = \sum_k \tr(F_k X) R_k.$$
Then the corresponding Choi matrix is 
\begin{equation}\label{Choi3}
C_\Phi = \sum_k F_k^t \otimes R_k.
\end{equation}
\end{lemma}

\begin{proof} This follows from \cite[Theorem 2 and Lemma 5]{Stormer}, or directly from the definition of the Choi matrix:
\begin{align*}
C_\Phi = \sum_{ij} E_{ij} \otimes \Phi(E_{ij}) 
&= \sum_{ij} E_{ij} \otimes \sum_k \tr(E_{ij} F_k) R_k\cr
&= \sum_k \left(\sum_{ij} \tr(E_{ij} F_k) E_{ij} \right)\otimes R_k\cr
&= \sum_k \left(\sum_{ij} \tr(E_{ji} F_k^t) E_{ij} \right)\otimes R_k\cr
&= \sum_k F_k^t \otimes R_k,
\end{align*}
where the final equality follows from the fact that the matrix units $\{E_{ij}\}$ are an orthonormal basis for $M_m$ with respect to the Hilbert-Schmidt inner product.
\end{proof}

\begin{theorem} \label{CQ} Let $0 \le T \in M_m \otimes M_n$.  
\begin{enumerate}
\item[(i)] $T$ is the Choi matrix for a QC channel iff $T$ is $B$-orthogonal with $\tr_B T = I$. 
\item[(ii)] $T$ is the Choi matrix for a CQ channel iff $T$ is $A$-orthogonal with $\tr_B T = I$.
\end{enumerate}
\end{theorem}

\begin{proof} (i) Let $\Phi:M_m \to M_n$ be a QC channel with Choi matrix $T$.  By definition, there is a Holevo representation \eqref{Holevo} with $R_1, \ldots, R_n$ rank one projections with sum $I_n$. By Lemma \ref{Choi} the Choi matrix for $\Phi$ is $$T = \sum_{k=1}^n F_k^t \otimes R_k.$$
Since $\sum_i R_i = I_n$, then $R_1, \ldots, R_n$ are orthgonal, so $T$ is  $B$-orthogonal. 
Since $\Phi$ is a quantum channel, then $\tr_B T = I$.

Conversely, suppose $T$ is $B$-orthogonal with $\tr_B T = I$ and $\rank \tr_A T = n$. Since $T \ge 0$, then $\Phi$ is completely positive, and since $\tr_B T= I$, then $T$ is trace preserving, so $T$ is a quantum channel.  By definition of $B$-orthogonality, we can write
$$T = \sum_{k=1}^p A_k \otimes B_k$$
with $0 \le A_1, \ldots, A_p$ and $0 \le B_1, \ldots, B_p$ with $B_1, \ldots, B_p$ orthogonal.  Via its spectral decomposition, we replace each $B_j$ by a linear combination of orthogonal rank one projections, and absorb scalar factors into the $A_j$'s.  Then we can write
\begin{equation}\label{Holevo2}
T = \sum_{j=1}^q F_k^t \otimes R_k
\end{equation}
with $R_1, \ldots, R_q$ orthogonal rank one projections. Clearly $q \le n$.  If $q < n$, we can define $F_{q+1}, \ldots, F_n$ to be zero, and choose rank one  projections $R_{q+1}, \ldots, R_n$ so that $\sum_i R_i = I_n$.  Thus $\Phi$ admits a Holevo form \eqref{Holevo2} in which $R_1, \ldots, R_n$ are rank one projections with sum $I_n$, so $\Phi$ is a QC channel.

The proof of the characterization of $CQ$ channels is similar.

\end{proof}

\section*{Faces of the separable state space}

A \emph{face} of a convex set $C$ is a subset $F$ such that if $A$ and $B$ are points in $C$ and a convex combination $tA+(1-t)B$ with $0 < t < 1$ is in $F$, then $A$ and $B$ are in $F$.  The intersection of faces is always a face, so for each point $A \in C$ there is a smallest face of $C$ containing $A$, denoted $\face_C A$. 

We let $K$ (or $K_d$) denote the convex set of states on $M_d$, i.e., the density matrices, and $S$ (or $S_{mn}$) denotes the convex set of separable states on $M_m \otimes M_n$. 
There is a canonical 1-1 correspondence between subspaces of $\C^d$ and faces of the state space $K_d$.  If $H$ is a subspace of $\C^d$ and $P$ is the projection onto $H$, then the associated face of $K_d$ is
\begin{equation}\label{F sub P} 
F_P = \{A \in K_d \mid \im A \subset \im P\} = \{A \in K_d \mid \im A \subset H\}.
\end{equation}
This correspondence of subspaces of $\C^d$ and faces of $K_d$ follows from, e.g.,  \cite[eqn. (3.14)]{Alfsen-Shultz}, which says that
$$F_P = \{A \in K_d \mid A = PAP\}.$$
By \eqref{range} this is equivalent to \eqref{F sub P}. (Equation (3.14) of \cite{Alfsen-Shultz} is stated in terms of positive linear functionals $\rho$ on $M_d$ associated with the density matrices $A$ in $M_d$ via $\rho(X) = \tr(AX)$, but it translates easily to \eqref{F sub P} above.)

From this it follows that faces of the state space of $M_m \otimes M_n$ are themselves ``mini state-spaces", i.e., are affinely isomorphic to some $K_p$ for $p\le mn$.  The extreme points of $K$ are precisely the pure states $P_x$,  where $P_x$ denotes the projection onto the span of the unit vector $x$. 

We recall for use below that the separable state space $S$ is compact, as is any face (since faces of closed finite dimensional convex sets are always closed.) The extreme points of $S$ are precisely the pure product states $P_{x \otimes y}$. 

 We now prove that certain faces of the separable state space are themselves ``mini separable state spaces", i.e., are affinely isomorphic to the separable state space $S_{pq}$ of $M_p \otimes M_q$ for some $p \le m, q \le n$.

\medskip

\begin{notation} If $V, W$ are subspaces of $\C^m$, $\C^n$ respectively with $\dim V = p$, $\dim W = q$, then $\Sep(V \otimes W) $ denotes the separable states in $M_m \otimes M_n$ that live on $V \otimes W$ (i.e., whose image is contained in $V \otimes W$). Note $\Sep(V \otimes W)$ is affinely isomorphic to the separable state space $S_{pq}$.
\end{notation}

We will make frequent use of the following implication for subspaces $V \subset \C^m, W \subset \C^n$:
$$ \text{for $x \in \C^m, y \in \C^n$},\quad x \otimes y \in V \otimes W \implies x \in V \text{ and } y \in W,$$
which follows immediately  by expanding bases of $V$ and $W$ to bases of $\C^m$ and $\C^n$ and expressing $x$ and $y$ in terms of these bases.  (Alternatively, cf.  \cite[eqn. (1.7)]{Greub}.

\begin{lemma} \label{product faces}Let $A\in M_m$, $B\in M_n$ be density matrices.  Then
$$\face_S(A \otimes B) = \Sep(\im A \otimes \im B).$$
\end{lemma}

\begin{proof}  Note that both sides are compact convex sets, and hence are the convex hull of their extreme points. The extreme points of both sides will be pure product states, so  we can restrict consideration to such states.

Suppose $P_{x \otimes y} \in \face_S(A \otimes B)$. This is contained in $\face_K(A \otimes B)$, which consists of the density matrices whose images are contained in $\im(A \otimes B) = \im A \otimes \im B$. Thus $x \in \im A$ and $y \in \im B$, so $P_{x \otimes y} \in \Sep(\im A \otimes \im B)$. Thus we shown
$$\face_S(A \otimes B) \subset \Sep(\im A \otimes \im B).$$

For the opposite inclusion, suppose $P_{x \otimes y}$ is any extreme point of  $\Sep(\im A \otimes \im B)$. Then $x \otimes y \in \im A \otimes \im B$ implies that $x \in \im A$ and $y \in \im B$. Hence $P_x$ is in $\face_K(A)$ and $P_y \in \face_K(B)$, so there exists a scalar $\lambda > 0$ such that $\lambda P_x \le A$ and $\lambda P_y \le B$. Then
$$A \otimes B = [(A - \lambda P_x) + \lambda P_x] \otimes [(B - \lambda P_y) + \lambda P_y]$$
Expanding the right sides gives four separable (unnormalized) states, and hence 
$$P_x \otimes P_y = P_{x \otimes y} \in \face_S(A \otimes B).$$
Thus
$$\Sep(\im A \otimes \im B)\subset \face_S(A \otimes B),$$
which completes the proof of the lemma.
\end{proof}

\begin{lemma}\label{sum} Let $V_1, \ldots, V_q$  be subspaces of $\C^m$ and $W_1, W_2, \ldots, W_q$   independent subspaces of $\C^n$. If $0 \not=x \otimes y \in \C^m \otimes \C^n$, let $J= \{\gamma \mid x \in V_\gamma\}$. Then
\begin{equation} \label{xy} x \otimes y \in \sum_{\gamma=1}^q V_\gamma \otimes W_\gamma
\end{equation}
iff $J$ is nonempty and $y \in \sum_{\gamma \in J} W_\gamma.$
\end{lemma}

\begin{proof} Assume \eqref{xy} holds. Then
$$x \otimes y \in \sum_{\gamma=1}^q V_\gamma \otimes W_\gamma \subset \left(\sum_\gamma V_\gamma\right) \otimes \left(\sum_\gamma W_\gamma\right)$$
so  $x \in \sum_\gamma V_\gamma$ and $y \in \sum_\gamma W_\gamma$.  Thus without loss of  generality we may assume $\sum_\gamma V_\gamma = \C^m$ and $\sum_\gamma W_\gamma = \C^n$.

Let $P_1, \ldots, P_q$ be the (non-self-adjoint) projection maps corresponding to the linear direct sum decomposition $\C^n = W_1 \oplus \cdots \oplus W_q$. Then for $1 \le \beta \le q$
\begin{equation}\label{eq1}
x\otimes P_\beta y = (I \otimes P_\beta)(x \otimes y) \in V_\beta \otimes W_\beta.
\end{equation}
If we choose $\beta$ so that $P_\beta y \not= 0$, then $x \in V_\beta$, so $J$ is not empty. Then for $\gamma \notin J$, we have $x \notin V_\gamma$, so by \eqref{eq1}, $P_\gamma y = 0$. It follows that $y \in \sum_{\gamma \in J} W_\gamma$.

Conversely, suppose $J$ is nonempty and $y \in \sum_{\gamma \in J} W_\gamma$,
say $y = \sum_{\gamma \in J} y_\gamma$.
Then
$$x \otimes y = \sum_{\gamma \in J} x\otimes y_\gamma \in \sum_{\gamma =1}^qV_\gamma \otimes W_\gamma.$$

\end{proof}

We say the convex hull  of a collection of convex sets $\{C_\alpha\}$ is a \emph{direct convex sum} if each point $x$ in the convex hull has a unique convex decomposition $x = \sum_\alpha \lambda_\alpha x_\alpha$ with $x_\alpha \in C_\alpha$. In the theorem below, $\co \,\bigoplus$ denotes the direct convex sum.

\begin{theorem}\label{faces} Let $T = \sum_\gamma A_\gamma \otimes B_\gamma$ be a density matrix in $M_m \otimes M_n$. Assume that $A_1, \ldots, A_p$ are density matrices with pairwise disjoint ranges,  and that $B_1, \ldots, B_p$ are positive matrices with independent images.  Then the face of the separable state space $S$ generated by $T$ is the direct convex sum 
\begin{equation}
\label{disjoint}
\face_S T = \co\,{\bigoplus_{\gamma=1}^p} \face_S(A_\gamma \otimes B_\gamma).
\end{equation}

\end{theorem}

\begin{proof} We first show that the convex hull on the right side of \eqref{disjoint} is a direct convex sum. First note that by the assumption that the images of the $B_\gamma$ are independent, it follows that the subspaces $\im A_\gamma \otimes \im B_\gamma$ are independent. (Indeed, combining product bases of $\im A_1 \otimes \im B_1, \ldots, \im A_p \otimes \im B_p$  gives a  basis of $\sum_\gamma \im A_\gamma \otimes \im B_\gamma$, from which the independence claim follows.)


Now suppose $C_\gamma, D_\gamma \in \face_S(A_\gamma \otimes B_\gamma)$ for $1 \le \gamma \le p$, and that 
\begin{equation} \label{direct}
\sum_\gamma C_\gamma = \sum_\gamma D_\gamma.
\end{equation}  
Then for any $\xi \in \C^m \otimes \C^n$
$$\sum_\gamma C_\gamma \xi = \sum_\gamma D_\gamma\xi.$$
Since 
$$\face_S(A_\gamma \otimes B_\gamma) \subset \face_K(A_\gamma \otimes B_\gamma)= \{E \in K \mid \im E \subset \im (A_\gamma \otimes B_\gamma)\},$$
then for each $\gamma$, $C_\gamma\xi$ and $D_\gamma\xi$ are in $\im (A_\gamma \otimes \B_\gamma)= \im A_\gamma \otimes \im B_\gamma$.  Hence by independence of the subspaces $\im A_\gamma \otimes \im B_\gamma$ , 
we  must have $C_\gamma \xi = D_\gamma\xi$ for each $\gamma$ and each vector $\xi$. Therefore $C_\gamma = D_\gamma$ for all $\gamma$, showing that the convex hull is indeed a direct convex sum.

Next we prove the equality in \eqref{disjoint}. Suppose $P_{x \otimes y}$ is in the left side.  Since the face that the state $P_{x \otimes y}$ generates in $S$ is contained in the face this state generates in $K$, then $x \otimes y $ is contained in the image of $ \sum_\gamma A_\gamma \otimes B_\gamma$, which is $\sum_\gamma \im A_\gamma \otimes \im B_\gamma$ (cf. \eqref{image of sum}). Since by assumption $A_1, \ldots, A_p$ are disjoint, the set $J$ in  Lemma \ref{sum} is a singleton set, so there is some $\beta$ such that $x \in \im A_\beta$ and $y \in \im B_\beta$. Then by Lemma \ref{product faces}, $P_{x \otimes y} \in \face_S(A_\beta \otimes B_\beta)$, which shows the left side of \eqref{disjoint} is contained in the right.

The extreme points of the right side are each contained in some $\face_S(A_\beta \otimes B_\beta)$, and since $A_\beta \otimes B_\beta$ is one of the summands on the left, then
$$\face_S(A_\beta \otimes B_\beta) \subset \face_S(\sum_\gamma A_\gamma \otimes B_\gamma),$$
which completes the proof of \eqref{disjoint}.

\end{proof}

\begin{lemma}\label{faces3} If $x$ is a unit vector in $\C^m$ and $B$ is a density matrix in $M_n$, then 
\begin{equation}\label{faces equal}
\face_S(P_x \otimes B) = \face_{K_{mn}}(P_x \otimes B) = P_x \otimes \face_{K_n} B.
\end{equation}
\end{lemma}

\begin{proof}  By Lemma \ref{product faces}, 
$$\face_S(P_x \otimes B) = \Sep(\im P_x \otimes \im B) = \Sep(\C x \otimes \im B).$$ 
Since every vector in $\C x \otimes \im B$ is a product vector, every density matrix whose image is contained in $\C x \otimes \im B$ is separable, as can be seen from   its spectral decomposition.  Thus by \eqref{F sub P} 
$$\Sep(\im P_x \otimes \im B) = \{E \in K_{mn} \mid \im E \subset  \im (P_x \otimes B)\} = \face_{K_{mn}} (P_x \otimes B),$$
 so the first equality of \eqref{faces equal} follows.

 Now we prove the second equality of \eqref{faces equal}.  If $T = P_x \otimes A$ with $A \in \face_{K_n}B$, then 
 $$\im T = \im (P_x \otimes \im A) \subset \im (P_x \otimes B)$$
 implies that $T \in \face_{K_{mn}}(P_x \otimes B) $, so we've shown
  $$P_x \otimes \face_{K_n} B \subset \face_{K_{mn}}(P_x \otimes B). $$
  To prove the reverse inclusion, let $T \in \face_{K_{mn}} (P_x \otimes B)$. Then $\im T \subset \C x \otimes \im B$.  We will prove there exists $A \in \face_{K_n}B$ such that $T = P_x \otimes A$, which will complete the proof of the lemma. 
 
 Since $\im T \subset \C x \otimes \im B$, for each $y \in \C^n$ there exists a unique $w \in \im B$ such that $T(x \otimes y) = x \otimes w$. Define $A \in M_n$ by $x \otimes Ay = T(x \otimes y)$ for $y \in \C^n$, and observe that $\im A \subset \im B$. 
 
For $z \in \C^m$, if  $z = x $ or $z \perp x$ we have $T(z \otimes y) = (P_x \otimes A)(z \otimes y)$.  It follows that $T = P_x \otimes A$. Since $T$ is a density matrix, it follows that $A$ also is a density matrix.  Since $\im A \subset \im B$, then $A \in \face_{K_n}B$.   Thus $T \in P_x \otimes \face_{K_n}B$.

\end{proof}

In Theorem \ref{faces}, the faces of the separable state space are expressed in terms of other (smaller) separable state spaces. In some circumstances, these are actually state spaces of the full matrix algebras, as we now show.  (This generalizes \cite[Thm. 4]{JMP}.)

\begin{theorem}\label{faces4} Let $T = \sum_{\gamma=1}^p A_\gamma \otimes B_\gamma$ be a density matrix in $M_m \otimes M_n$. Assume that $A_1, \ldots, A_p$ are rank one density matrices,  and that $B_1, \ldots, B_p$ are positive matrices with independent images.  Then there are unit vectors $x_1, \ldots x_q$, with $P_{x_1}, \ldots, P_{x_q}$ distinct, and independent density matrices $C_1, \ldots, C_q$ in $M_m$, such that $T$ admits the convex decomposition
\begin{equation}
T = \sum_{\nu=1}^q \lambda_\nu P_{x_\nu} \otimes C_\nu.
\end{equation}
 This decomposition is unique, and the face of $S$ generated by each $P_{x_\nu} \otimes C_\nu$ is also a face of $K_{mn}$, so that
\begin{equation}
\label{disjoint2}
\face_S T = {\co}\,{\bigoplus_{\nu=1}^q} \face_{K_{mn}}(P_{x_\nu} \otimes C_\nu)
= {\co}\, {\bigoplus_{\nu=1}^q} (P_{x_\nu} \otimes \face_{K_n} C_\nu).
\end{equation}
\end{theorem}

\begin{proof} By assumption, each $A_\gamma$ is a postive scalar multiple of a projection $P_{x_\gamma}$, where $x_\gamma$ is a unit vector in $\C^m$.  Absorbing this scalar into $B_\gamma$, we write the given decomposition in the form
$$T = \sum_{\gamma= 1}^p P_{x_\gamma} \otimes \widetilde B_\gamma.$$
Now we collect together terms where the first factors $P_{x_\gamma}$ coincide.  In precise terms, we define an equivalence relation on the indices $\{1, \ldots, p\}$ by $\gamma \sim \kappa$ if $\C x_\gamma = \C x_\kappa$, or equivalently if $P_{x_\gamma} = P_{x_\kappa}$. Let $J$ be the set of equivalence classes, and for each equivalence class $\nu \in J$ choose a representative $\gamma \in \nu$ and define $\widetilde x_\nu = x_\gamma$. 
Then
 $$T = \sum_{\nu \in J} P_{\widetilde x_\nu} \otimes C_\nu$$
 where $C_\nu = \sum_{\gamma \in \nu} B_\gamma$.
 Define $q= |J| = q$; numbering the members of $J$ in sequence gives a decomposition of the form specified in the theorem.

Since the images of the $P_{\widetilde x_\nu}$ are disjoint, the final statement of the theorem   follows from Theorem \ref{faces} and Lemma \ref{faces3}.

\end{proof}


\section*{Decompositions into pure product states}

If $T$ is $B$-independent, Theorem \ref{main} provides a canonical way to decompose $T$.  Then with the notation of Theorem \ref{main}, we can decompose each $A_\gamma$ and $B_\gamma$ further via the spectral theorem into linear combinations of rank one projections, and  this gives a representation of $T$ as a convex combination of pure product states. The next result describes when this decomposition into pure product states is unique, generalizing  the uniqueness result in \cite[Corollary 5]{JMP}.

\begin{theorem}\label{unique} If $T \in M_m \otimes M_n$ is a $B$-independent density matrix, then there is a unique decomposition of $T$ as a convex combination of pure product states iff in the canonical decomposition \eqref{T2} of Theorem \ref{main}, each $A_\gamma$ and each $B_\gamma$ has rank one.  Thus the decomposition of $T$ into pure product states is unique iff
$T$ can be written as a convex combination
$$T = \sum_{\gamma=1}^p \lambda_\gamma P_{x_\gamma} \otimes P_{y_\gamma}$$
with unit vectors $y_1, \ldots, y_p$ that are linearly independent, and unit vectors $x_1, \ldots, x_p$ such that $P_{x_1}, \ldots, P_{x_p}$ are distinct.

\end{theorem}

\begin{proof} Suppose that $T$ is $B$-independent and admits a unique decomposition as a convex combination of pure product states. Let $T = \sum_\gamma A_\gamma \otimes B_\gamma$ be the canonical decomposition of $T$ given in Theorem \ref{main}.  If any $A_\beta$ does not have rank one, then there are infinitely many ways to write $A_\beta$ as a convex combination of pure states, which when combined with any decomposition into rank one projections for  the other $A_\gamma$ and each $B_\gamma$ gives infinitely many decompositions of $T$ into pure product states.  The same argument applies if any $B_\gamma$ does not have rank one. Hence if $T$ admits a unique convex decomposition into pure product states, each $A_\gamma$ and $B_\gamma$ must have rank one.

Conversely, assume that $T$ can be written as a convex combination
 $$T = \sum_\gamma \lambda_\gamma P_{x_\gamma}\otimes P_{y_\gamma}$$
with $\{P_{x_\gamma}\}$ distinct and with $\{y_\gamma\}$ independent.  
This decomposition satisfies the hypotheses of Theorem \ref{faces4}, and thus  $\face_S(T)$ will be the direct convex sum of the singleton faces $\{P_{x_\gamma} \otimes P_{y_\gamma}\}$. 
 Now suppose that we are given any other convex decomposition into pure product states
$$T = \sum_\nu t_\nu P_{z_\nu} \otimes P_{w_\nu},$$
where we are not making any assumption about independence of $\{P_{w_\nu}\}$ or distinctness of $\{P_{z_\nu}\}$. Then each  $P_{z_\nu} \otimes P_{w_\nu}$ is in $\face_ST$ and is an extreme point of the separable state space $S$. By the definition of a direct convex sum, we conclude that each $P_{z_\nu} \otimes P_{w_\nu}$ must coincide with some $P_{x_\gamma} \otimes P_{y_\gamma}$. Thus the convex decomposition of $T$ into pure product states is unique.
\end{proof}

\begin{remark} One might suspect that for the uniqueness conclusion in Theorem \ref{unique}, it would  suffice for the joint eigenspaces of the $(\wT)_{ij}$ to be one dimensional, but this is not correct, as can be seen by considering $A \otimes P_y$ where $\rank A > 1$.
\end{remark}

\section*{The marginal rank condition}

In this section we specialize previous results to an important class of separable states.

\begin{definition} A density matrix $T \in M_m \otimes M_n$ satisfies the \emph{marginal rank condition} if $\rank T = \max(\rank T_A,\rank T_B)$,
which reduces to $\rank T = \rank T_B$ if $m \le n$, which we will assume in the sequel.
\end{definition}

 
%
%
%
%
%


We will see that such matrices, if separable, are $B$-independent.
The following lemma for states on $M_m \otimes M_n$ appeared for $m = 2$ in \cite{Kraus}, and for general $m, n$   in \cite[Lemma 6, and proof of Theorem 1]{HLVC}.  An alternate shorter proof for general $m, n$  can be found in \cite[Lemma 13]{RuskaiWerner}. It shows that separable density matrices satisfying the marginal rank condition are the same as those that admit a representation \eqref{sep def} with each $A_\gamma$ and $B_\gamma$ of rank one, and with $B_1, \ldots, B_p$ independent.

\begin{lemma} \label{rank lemma} Let $T$ be separable.  Then $T$ admits a decomposition $T = \sum_{i=1}^p \lambda_i P_{x_i \otimes y_i}$ with $y_1, \ldots, y_p$ independent iff $\rank T = \rank T_B$.  \end{lemma}

We now show that Theorem \ref{main} gives a practical way to check  whether a particular matrix satisfying the marginal rank condition is separable. Theorem \ref{main} then also provides a way to find an explicit representation of $T$ as a convex combination of tensor products of positive matrices. (For testing separabiltiy, the PPT test also suffices, cf. \cite{HLVC}.)

\begin{theorem} Let $T\in M_m \otimes M_n$ with $\rank T = \rank T_B$. Define $\wT$ as in \eqref{wT def}. Then $T$ is separable iff the matrices $(\wT)_{ij}$ are normal and commute. 
\end{theorem}

\begin{proof}  Assume $T$ has marginal rank. If $T$ is separable then by Lemma \ref{rank lemma}, $T$ is $B$-independent, and hence by Theorem \ref{main}, the matrices $(\wT)_{ij}$ are normal and commute. Conversely, if these matrices are normal and commute, then by Theorem \ref{main} $T$ is $B$-independent, and hence separable.

\end{proof}

\begin{corollary}\label{unique decomp} If $T$ is separable of marginal rank, then $T$ admits a unique decomposition
$$T = \sum_{\gamma=1}^p P_{x_\gamma} \otimes B_\gamma,$$
with $B_1, \ldots, B_p$ positive matrices with independent images, and $x_1, \ldots, x_p$ unit vectors with $P_{x_1}, \ldots, P_{x_p}$ distinct.\end{corollary}

\begin{proof} By Theorem \ref{main}, there is a unique decomposition
$$T = \sum_\gamma A_\gamma \otimes B_\gamma$$
where each $A_\gamma$ is positive with trace 1, $A_1, \ldots, A_p$ are distinct, and $B_, \ldots, B_p$ are positive and have independent images. By the marginal rank condition,
$$\rank T = \sum_\gamma \rank A_\gamma \rank B_\gamma = \rank T_B = \sum_\gamma \rank B_\gamma$$
which implies that $\rank A_\gamma = 1$ for all $\gamma$. Thus there are unit vectors $x_\gamma$ such that $A_\gamma  = P_{x_\gamma}$. Distinctness of $A_1, \ldots, A_p$ implies that $P_{x_1}, \ldots, P_{x_p}$ distinct.
\end{proof}

\begin{remark}  Uniqueness of the decomposition in Corollary \ref{unique decomp} was first proved in \cite[Section IV]{JMP}. Corollary \ref{unique decomp}   provides an alternate proof, and Theorem \ref{main} provides an explicit way to find that decomposition. 
\end{remark}

\section*{Summary}

We have defined a class of separable states ($B$-independent states) which generalizes the separable states whose rank equals their marginal rank. We have described an intrinsic way to check membership in this class, and for such states we have given a procedure leading to a unique canonical separable decomposition into  product states.

\end{document}